\documentclass[11pt]{amsart}
\usepackage{amsmath, amsfonts, amsthm, amssymb, mathrsfs}
\usepackage{paralist}
\usepackage{fullpage}
\usepackage{anyfontsize}
\usepackage{multicol}
\usepackage{xcolor}
\usepackage{url}
\usepackage{pdfpages}
\usepackage{tikz} 
\usepackage{graphicx}
\usepackage[foot]{amsaddr}
\usetikzlibrary{shapes,arrows, patterns,shapes.geometric}
\tikzstyle{b} = [draw, thick, black, -]

\newtheorem{theorem}{Theorem}[section]

\theoremstyle{definition}

\newtheorem{remark}[theorem]{Remark}

\newtheorem{thm}{Theorem}
\newtheorem{co}[thm]{Corollary}
\newtheorem{lem}[thm]{Lemma}

\newtheorem{assumption}[thm]{Assumption}

\newtheorem{pr}[thm]{Proposition}

\newtheorem{defi}[thm]{Definition}
\newtheorem{exmp}[thm]{Example}

\newcommand{\F}{\mathbb{F}}
\newcommand{\Z}{\mathbb{Z}}
\newcommand{\R}{\mathbb{R}}
\newcommand{\Q}{\mathbb{Q}}
\newcommand{\N}{\mathbb{N}}

\newcommand{\row}{\text{Row}}
\newcommand{\poly}{\text{poly}}
\newcommand{\conv}{\text{conv}}
\newcommand{\supp}{\text{Supp}}

\usepackage[colorlinks=true, allcolors=blue]{hyperref}

\title{Pseudocodewords of quantum, quasi-cyclic, and spatially-coupled LDPC codes: a fundamental cone perspective}

\author{Wittawat Kositwattanarerk${}^1$}
\author{Gretchen L.\ Matthews${}^2$}
\author{Emily McMillon${}^2$}
\author{Tunchanok Yutitumsatit${}^1$}

\address{${}^1$Department of Mathematics, Faculty of Science, Mahidol University, Bangkok, 10400, Thailand}
\address{${}^2$Department of Mathematics, Virginia Tech, Blacksburg, 24061, Virginia, USA}

\thanks{Matthews was partially supported by NSF DMS-2201075 and the Commonwealth Cyber Initiative.}
\thanks{McMillon was supported by NSF DMS-2303380 and partially supported by the Commonwealth Cyber Initiative.}



\begin{document}

\begin{abstract}
While low-density parity-check (LDPC) codes are near capacity-achieving when paired with iterative decoders, these decoders may not output a codeword due to the existence of pseudocodewords. Thus, pseudocodewords have been studied to give insight into the performance of modern decoders including iterative and linear programming decoders. These pseudocodewords are found to be dependent on the parity-check matrix of the code and the particular decoding algorithm used. In this paper, we consider LP decoding, which has been linked to graph cover decoding, providing functions which capture these pseudocodewords. In particular, we analyze the underlying structure of pseudocodewords from quantum stabilizer codes that arise from LP decoding, quasi-cyclic LDPC codes, and spatially-coupled LDPC codes.
\end{abstract}

\maketitle

\section{Introduction}
\label{sec:intro}
Graph-based iterative decoders have been extensively studied due to their decoding efficiency and performance when used in conjunction with low-density parity-check (LDPC) codes; see for instance \cite{10.5555/3137384,1222728}. LDPC codes for quantum error correction (Q-LDPC codes) were first introduced in \cite{MMM} and are shown to have remarkable properties including facilitating fault-tolerant quantum computation with constant overhead \cite{Breuckmann,Gottesman,xu2023constant}. Q-LDPC codes have been studied extensively; see \cite{EKZ_22,krishna2023vidermans,LZ_22,Panteleev_22,TZ_14} for instance. Surface codes \cite{IBM}, which were introduced in \cite{bravyi1998quantum}, are among the Q-LDPC codes providing further motivation to study Q-LDPC codes. Indeed, there is much to be understood about how their representation impacts performance, a question that is still of interest for their classical counterpart. There are other parallels with the classical case. For instance, while the framework for iterative decoding applies to any parity-check code, descriptions of these more general codes (meaning beyond those with sparse matrix representations) make their implementation impractical. Similarly, while the theory we consider in this paper also applies to CSS codes that do not arise from LDPC codes, its usefulness may be hindered by complicated descriptions of the codes in terms of their parity-check matrices.

Quasi-cyclic LDPC (QC-LDPC) codes are a widely used class of LDPC codes due to their simple implementations. They have been used to analyze the pseudocodewords from LDPC convolutional codes \cite{4957655}, and, more recently, quasi-cyclic moderate-density parity-check (QC-MDPC) codes have become more relevant due to their role in post-quantum cryptography schemes based on the McEliece cryptosystem \cite{MTSB13} such as Bit Flipping Key Encapsulation (BIKE) \cite{BIKE}. 

The fundamental cone was introduced in \cite{KV} as the convex hull of all graph cover pseudocodewords and is proven useful in the study of pseudocodewords \cite{KM,KLVW}. In this paper, we consider the fundamental cone and associated pseudocodewords to capture impediments to successful graph-based iterative decoding, especially for the linear programming (LP) decoder \cite{FWK}. We build on the recent works of \cite{FGL,Li_Vontobel_ISIT} to consider the fundamental cone for quantum stabilizer codes and associated pseudocodewords. LP decoding for Q-LDPC codes has been introduced and examined in e.g. \cite{BBANH,groues:tel-04018968,javed2024lowcomplexity,Li_Vontobel_ISIT}.
The approach considered in this paper applies to LDPC codes and the MDPC codes considered for cryptographic use, and hence applies to QC-LDPC, QC-MDPC, and Q-LDPC codes, providing insight into LP and graph cover decoding among other message passing algorithms. Our approach differs from the previous work on pseudocodewords by \cite{Li_color,LVgraph,Rox_ISIT_09,Rox_Trans_12} due to its focus on generating functions and their concise representation of pseudocodewords. Analysis of the pseudocodewords will also be helpful toward pseudocodeword-based decoding \cite{Li_color,LVgraph}.

This paper is organized as follows. Section \ref{sec:preliminaries} contains background information on codes, LP decoding, pseudocodewords, and the fundamental cone. Section \ref{sec:fund_cone} includes general results on the fundamental cone of binary parity-check codes. Section~\ref{sec:classes} applies results given in Section~\ref{sec:fund_cone} to three specific sub-classes of parity-check codes: quantum codes in Section~\ref{sec:quantum}, quasi-cyclic codes in Section~\ref{sec:qc}, and spatially-coupled codes in Section~\ref{sec:sc}. Section \ref{sec:conclusions} ends the paper with a summary and open problems.

\section{Preliminaries} \label{sec:preliminaries}

We begin by introducing some standard notation we will use throughout. The set of nonnegative integers is denoted by $\N$. Given a positive integer $t$, $[t]:=\left\{ 1, \dots, t \right\}$. $\F_2$ denotes the finite field with two elements, and $\F_2^{m \times n}$ denotes the set of $m \times n$ matrices with entries from $\F_2$. The all-zero matrix of size $m \times n$ is denoted by ${0}_{m \times n}$. For $j \in [m]$, $Row_j(A)$ denotes the $j^{th}$ row of the matrix $A \in \F_2^{m \times n}$. Vectors are denoted in boldface, i.e. $\mathbf{x}$. The support of $\mathbf{x} \in \F_2^n$ is $\supp({\mathbf x})=\left\{ i \in [n]: x_i \neq 0 \right\}$. A graph is a pair $(V;E)$ where $V$ is the set of vertices and $E$ is the set of edges. Vertices $v_1,v_2\in V$ are adjacent, or connected by an edge, if $(v_1,v_2)\in E$.

The codes we consider are binary linear codes, which we now define.

\begin{defi}
Let $k$ and $n$ be positive integers where $k \leq n$. A \textit{binary linear code} $C$ of length $n$ and dimension $k$ is a subspace of $\F_2^n$ of dimension $k$. The \textit{minimum distance} $d$ of $C$ is the minimum support size of a nonzero vector in $C$, i.e.
\[d=\min\{|\supp(\mathbf{c})|\mid \mathbf{c}\in C\setminus\{\mathbf{0}\}\}.\]
If $n$, $k$, and $d$ are known for a code $C$, we call $C$ an $[n,k,d]$ code.

A matrix $H\in\mathbb{F}_2^{r\times{}n}$, where $r \geq n-k$, is called a \textit{parity-check matrix} of $C$ if $C$ is the null space of $H$. Since a parity-check matrix of a code is not unique, we write $C$ as $C(H)$ to specifically designate $C$ to the matrix $H$. This makes
\[ C(H) = \{\mathbf{c} \in \F_2^n : H \mathbf{c}^T = \mathbf{0} \}.\]
\end{defi}

The code $C(H)$ is called a \textit{low-density parity-check (LDPC) code} if $H \in \F_2^{m \times n}$ has row weight $O(1)$, meaning few non-zero entries. One may also consider \textit{moderate-density parity-check (MDPC) codes} $C(H)$ which are given by matrices $H \in \F_2^{m \times n}$ with row weight $O(\sqrt{n})$.

In this paper, we consider the LP decoder for both classical and quantum codes, which also gives insight into other iterative decoders, most notably graph-based iterative decoders. 
Viewing codewords as elements of $\R^n$, the codeword polytope of a binary linear code C is the convex hull of the codewords:
$$\poly(C) =
\left\{ \sum_{\mathbf{c} \in C} \lambda_{\mathbf{c}} \mathbf{c} : \lambda_\mathbf{c} \in [0,1], \sum_{\mathbf{c} \in C} \lambda_\mathbf{c} =1 \right\} \subseteq [0,1]^n,$$
which depends on the set of codewords (as opposed to the representation of the code). We note that the vertices of $\poly(C)$ are precisely the codewords of $C$.

Maximum likelihood (ML) decoding may be described using the codeword polytope as follows; for more details, see \cite{FWK}. Given a received word $\mathbf{w} \in \F_2^n$ and $i \in [n]$, let 
$$
\gamma_i:=- \log \left( \frac{Pr(w_i \mid 1)}{Pr(w_i \mid 0)} \right)
$$
denote the log likelihood ratio (LLR) associated with position $i$.
Then the ML decoding problem is to find $\mathbf{y} \in C$ such that the cost function $\sum_{i=1}^n \gamma_i y_i$ is minimized, which can be stated as 
\begin{equation}\label{E:ML_LP}
\min \sum_{i=1}^n \gamma_i y_i \textnormal{ subject to } y \in \poly(C).
\end{equation}
Generally, the solution to the LP in Equation (\ref{E:ML_LP}) will be a vertex of $\poly(C)$ and the vertices of $\poly(C)$ are precisely the codewords of $C$, so the solution will be a codeword of $C$ as desired. However, in general, the description of $\poly(C)$ is exponential in $n$ \cite{FWK}. It is this complexity that motivates the relaxed polytope, which in turn invites non-codeword pseudocodewords. 

To obtain a more tractable representation, the codeword polytope is relaxed by fixing a particular parity-check matrix \cite{FWK}. Let $H \in \F_2^{r \times n}$ be a parity-check matrix for the binary linear code $C(H)$. For $j \in [r]$, each row of $H$ defines a 
codeword polytope of a supercode of $C(H)$: $\poly( C(\row_j (H)))$. Noting that $C = \bigcap_{j=1}^r C \left( \row_j (H)\right)$, the relaxed polytope of $C(H)$ is 
$$R(H):=\bigcap_{j\in [r]} \poly(C(\row_j (H))) \subseteq [0,1]^n$$
as defined in \cite{FWK}. Unlike the codeword polytope $\poly(C)$, the relaxed polytope depends on the choice of parity-check matrix representation. The associated LP, sometimes called the Linear Code Linear Program, may be stated as 
\begin{equation}\label{E:LCLP}
\min \sum_{i=1}^n \gamma_i y_i \textnormal{ subject to } \mathbf{y} \in R(H).
\end{equation}
It has been shown that if the optimal solution $\mathbf{y}$ satisfies $\mathbf{y} \in \F_2^n$, then $\mathbf{y}$ is the ML codeword \cite[Proposition 2]{FWK}. Moreover, the description of the relaxed polytope $R(H)$ is polynomial in the code length $n$ \cite[Appendix II]{FWK}, rather than exponential as is the case for $\poly(C)$. 
However, an optimal solution to the LP in (\ref{E:LCLP}) is a vertex of $R(H)$ which might not be a vertex of $\poly(C(H))$, meaning a noncodeword. 
Indeed, note that in general 
$$
\poly(C(H)) \subsetneqq R(H)$$
inviting the possibility that noncodewords may appear as outputs of the linear program in (\ref{E:LCLP}) optimizing over $R(H)$. This situation motivates the following definition. 
\begin{defi} \label{def:pcw}
    An LP pseudocodeword of the code $C(H)$ is a vertex of the relaxed polytope $R(H)$.
\end{defi}

The original setting for LP decoding using the Linear Code Linear Program is classical. More recently, it was demonstrated in \cite{Li_Vontobel_ISIT} that there are analogous results for quantum codes through the use of label codes. Among them are the fact that this polytope is the correct one for blockwise MAP decoding in the setting they consider, yet it is intractable. More details will be provided in Section \ref{sec:quantum}.

Alternatively, pseudocodewords can be defined from a graph cover perspective and made precise with the definition of the fundamental cone. A graph cover is defined over a bipartite graph called \textit{Tanner graph}, which is used for iterative decoding algorithms on parity-check codes.

\begin{defi}
    Let $C(H)$ be the code defined by parity-check matrix $H \in \F_2^{m \times n}$. The \textit{Tanner graph} $T(H)$ of $C$ corresponding to $H$ is the bipartite graph $T(H) = (V,W;E)$ with vertex set $V$ corresponding to the $n$ columns of $H$ and vertex set $W$ corresponding to the $m$ rows of $H$ such that there is an edge between $v_i \in V$ and $w_j \in W$ if and only if $H_{j,i} = 1$.
\end{defi}

From a graphical perspective, the Tanner graph encodes codewords as follows: a vector is assigned to the vertex set $V$. This vector is in $C$ if and only if the sum at each node $w \in W$ of all assignments of the neighbors of $w$ in $V$ is equal to $0$ modulo $2$.

Iterative decoders based on message-passing algorithms act locally on a Tanner graph in order to find a global solution. This aspect of iterative decoders delivers a great strength -- speed -- at the cost of optimality. Loosely speaking, iterative graph decoding operates in the following way: a received message ($0$'s and $1$'s, or their likelihoods) from the channel is put onto the variable nodes (sometimes along with a reliability value). Each variable node then sends information to its neighboring check nodes. Each check node now calculates new estimates for the assignments to its neighboring variable nodes and sends this information back. This process continues until either a codeword is obtained or a certain predetermined number of iterations have occurred.

Generally speaking, these algorithms converge to a solution fairly quickly, but the structure of the Tanner graph can sometimes cause iterative decoders to fail or converge to non-optimal solutions. This is caused by the existence of \textit{graph cover pseudocodewords} which, loosely speaking, are solutions to some finite lift of a code's Tanner graph representation.

\begin{defi}
A finite graph cover of $T(H)$ is a bipartite graph $(V_1\cup V_2\cup\ldots\cup V_t,W_1\cup W_2\cup\ldots\cup W_t;E)$ where $V_k$ and $W_k$ are disjoint copies of $V$ and $W$ respectively for all $k\in[t]$. If there is an edge between $v_i \in V$ and $w_j \in W$ (i.e., $H_{j,i} = 1$.), then each copy of $v_i$ will be adjacent to exactly one copy of $w_j$.
\end{defi}

\begin{exmp} \label{ex:Tanner_graph}
    Let the following matrix be a parity-check matrix for a code $C(H)$.
    \[ H = \begin{bmatrix} 1 & 1 & 0 & 1 & 0 & 0 \\ 0 & 1 & 1 & 0 & 1 & 0 \\ 0 & 0 & 0 & 1 & 1 & 1 \end{bmatrix} \]
    Then the corresponding Tanner graph of the code, $T(H)$, is given in Figure~\ref{fig:Tanner_graph}. Check nodes, which correspond to rows in $H$, are denoted by squares, and variable nodes, which correspond to columns of $H$, are denoted by circles, as is convention. Additionally, Figure~\ref{fig:cover} gives an example of a graph cover of $T(H)$.
\end{exmp}

\begin{figure}[h]
\centering
\begin{tikzpicture}[scale=.5,square/.style={regular polygon,regular polygon sides=4}]

\node (v1) at (0,0) [circle,draw,fill=black,scale=0.5] {};
\node (v2) at (2,0) [circle,draw,fill=black,scale=0.5] {};
\node (v3) at (4,0) [circle,draw,fill=black,scale=0.5] {};
\node (v4) at (6,0) [circle,draw,fill=black,scale=0.5] {};
\node (v5) at (8,0) [circle,draw,fill=black,scale=0.5] {};
\node (v6) at (10,0) [circle,draw,fill=black,scale=0.5] {};
\node (c1) at (2,3) [square,draw,fill=black,scale=0.5] {};
\node (c2) at (5,3) [square,draw,fill=black,scale=0.5] {};
\node (c3) at (8,3) [square,draw,fill=black,scale=0.5] {};

\path[b] (v1) to (c1);
\path[b] (v2) to (c1);
\path[b] (v2) to (c2);
\path[b] (v3) to (c2);
\path[b] (v4) to (c1);
\path[b] (v4) to (c3);
\path[b] (v5) to (c2);
\path[b] (v5) to (c3);
\path[b] (v6) to (c3);

\end{tikzpicture}
\caption{The Tanner graph associated with the parity check matrix in Example~\ref{ex:Tanner_graph}.}
\label{fig:Tanner_graph}
\end{figure}
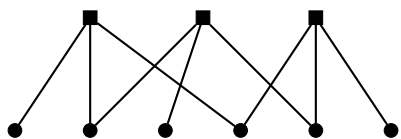

\begin{figure}[h]
\centering
\begin{tikzpicture}[scale=.5,square/.style={regular polygon,regular polygon sides=4}]

\node (v11) at (0,0) [circle,draw,fill=black,scale=0.5] {};
\node (v12) at (2,0) [circle,draw,fill=black,scale=0.5] {};
\node (v13) at (4,0) [circle,draw,fill=black,scale=0.5] {};
\node (v14) at (6,0) [circle,draw,fill=black,scale=0.5] {};
\node (v15) at (8,0) [circle,draw,fill=black,scale=0.5] {};
\node (v16) at (10,0) [circle,draw,fill=black,scale=0.5] {};
\node (c11) at (2,3) [square,draw,fill=black,scale=0.5] {};
\node (c12) at (5,3) [square,draw,fill=black,scale=0.5] {};
\node (c13) at (8,3) [square,draw,fill=black,scale=0.5] {};

\node (v21) at (0.5,0) [circle,draw,fill=black,scale=0.5] {};
\node (v22) at (2.5,0) [circle,draw,fill=black,scale=0.5] {};
\node (v23) at (4.5,0) [circle,draw,fill=black,scale=0.5] {};
\node (v24) at (6.5,0) [circle,draw,fill=black,scale=0.5] {};
\node (v25) at (8.5,0) [circle,draw,fill=black,scale=0.5] {};
\node (v26) at (10.5,0) [circle,draw,fill=black,scale=0.5] {};
\node (c21) at (2.5,3) [square,draw,fill=black,scale=0.5] {};
\node (c22) at (5.5,3) [square,draw,fill=black,scale=0.5] {};
\node (c23) at (8.5,3) [square,draw,fill=black,scale=0.5] {};

\path[b] (v11) to (c11);
\path[b] (v12) to (c21);
\path[b] (v12) to (c12);
\path[b] (v13) to (c12);
\path[b] (v14) to (c11);
\path[b] (v14) to (c13);
\path[b] (v15) to (c22);
\path[b] (v15) to (c13);
\path[b] (v16) to (c13);

\path[b] (v21) to (c21);
\path[b] (v22) to (c11);
\path[b] (v22) to (c22);
\path[b] (v23) to (c22);
\path[b] (v24) to (c21);
\path[b] (v24) to (c23);
\path[b] (v25) to (c12);
\path[b] (v25) to (c23);
\path[b] (v26) to (c23);

\end{tikzpicture}
\caption{A graph cover of $T(H)$ from Example~\ref{ex:Tanner_graph} where there is an extra copy of each vertex.}
\label{fig:cover}
\end{figure}
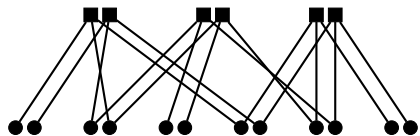

Now, a graph cover itself can act as a Tanner graph, and graph cover pseudocodewords are essentially codewords of these graph covers. For more background on graph cover pseudocodewords, we refer the interested reader to \cite{KLVW} or \cite{KV}. For a code $C(H)$ with parity-check matrix $H$, we denote the associated set of graph cover pseudocodewords by $\mathcal{P}(H)$. It was shown in \cite{KV} that rational elements of the relaxed polytope, $R(H) \cap \Q^n$, are scaled graph cover pseudocodewords. The associated fundamental cone has proven useful in the analysis of these outputs.

\begin{defi} \label{def:cone}
    The fundamental cone of $H \in \F_2^{r \times n}$ is 
$$
    \mathcal{K}(H) = \left\{ \mathbf{v} \in \R_{\geq 0}^n : \row_j (H) \cdot \mathbf{v} \geq 2 h_{ji}v_i \ \forall i \in [n], \forall j \in [r]\right\}. 
$$
\end{defi}

The relationship between the fundamental cone and the set of graph cover pseudocodewords is given in the following result.

\begin{thm} \cite[Theorem 4.4]{KLVW} \label{thm:equiv}
    Let $H \in \F_2^{r \times n}$ and $\mathbf{p} \in \N^n$. Then $\mathbf{p} \in \mathcal{P}(H)$ if and only if $\mathbf{p} \in \mathcal{K}(H)$ and $H \mathbf{p}^T = \mathbf{0} \mod 2$.
\end{thm}

Throughout, we drop the ``graph cover'' adjective on pseudocodewords when it is clear from the context. We will make use of the fundamental cone and this equivalence to give insights into graph cover pseudocodewords.

\section{Properties of the fundamental cone for general codes} \label{sec:fund_cone}

In this section, we develop the key properties of the fundamental cone that will be applied to three classes of parity-check codes in Section~\ref{sec:classes} and provide key observations on the fundamental cone which will be useful in studying the set $\mathcal \mathcal{P}(H)$ of pseudocodewords of a binary code $C(H)$. We will consider associated generating functions, a direction  motivated by the work of Koetter, Li, Vontobel, and Walker \cite{KLVW} concerning
cycle codes, meaning those in which all bit nodes in the associated Tanner graph $T(H)$ have degree $2$. Beginning in \cite{KLVW_ITW}, they associate a normal graph with the parity-check matrix for such a code and then consider the Hashimoto edge zeta function of this normal graph. 
This allows for the pseudocodewords of a cycle code to be read off from
the monomials appearing in the power series expansion of an associated rational function. This point of view was carried on in \cite{KM} and \cite{KM_2} to study the set of pseudocodewords via their generating functions. In particular, the set of pseudocodewords $\mathcal{P}(H)$ of a code $C(H)$ with parity-check matrix $H$ can be captured via the $n$-variate generating function $f_{H}$:
$$ f_{H}({\bf{x}}):=\sum_{\mathbf{p} \in \mathcal{P}(H)} {\bf{x}}^\mathbf{p} \in \F_2[x_1,\dots, x_n],$$
where $\mathbf{x}^{\mathbf{p}} := x_1^{p_1} \dotsm x_n^{p_n}$ for $\mathbf{p} \in \mathcal{P}(H)$ and $\mathbf{x} \in \F_2[x_1, \dotsc, x_n]$.

The following handful of results (Lemma~\ref{lem:intersection}, Theorem~\ref{thm:products}, and Theorem~\ref{thm:blockrow}) all consider natural ways of combining collections of block matrices into larger parity-check matrices.

\begin{lem} \label{lem:intersection}
Let $H_1,\ldots,H_t$ be a collection of matrices formed by the rows of $H$, where $H_\ell \in \F_2^{r_\ell \times n}$ and $r := \sum_{\ell=1}^t r_{\ell}$.
   The fundamental cone of
   \[ H = \begin{bmatrix} H_1 \\ H_2 \\ \vdots \\ H_t \end{bmatrix} \in\F_2^{r \times n}\]
  is
    \[ \mathcal{K}(H) = \bigcap_{\ell=1}^t \mathcal{K}(H_\ell).\]
\end{lem}

\begin{proof} The result follows almost immediately from the definition, since
    \begin{align*}
        \mathcal{K}(H) &= \left\{ \mathbf{v} \in \R_{\geq 0}^n : \row_j (H) \cdot \mathbf{v} \geq 2 h_{ji} v_i \ \forall i \in [n], \forall j \in [r]\right\} \\
        &= \bigcap_{\ell=1}^t \left\{ \mathbf{v} \in \R_{\geq 0}^n : \begin{array}{l} \row_j(H_\ell) \cdot \mathbf{v} \geq 2 (H_\ell)_{ji} v_i \\ \forall i \in [n], \forall j \in [r_\ell] \end{array} \right\} \\
        &= \bigcap_{\ell=1}^t \mathcal{K}(H_\ell).
    \end{align*}
\end{proof}

\begin{thm} \label{thm:products}
        Let 
    \[ H = \begin{bmatrix} H_1 & 0 & \dotsm & 0 \\
                           0 & H_2 &  & 0 \\
                           \vdots & & \ddots & \vdots\\
                           0 & 0 & \dotsm & H_t \end{bmatrix} \in \F_2^{r \times n}.\]
    Then:
    \begin{enumerate}
        \item $R(H) = R(H_1) \times \dotsm \times R(H_t)$;
        \item $\mathcal{K}(H) = \mathcal{K}(H_1) \times \dotsm \times \mathcal{K}(H_t)$;
        \item $\mathcal{P}(H) = \mathcal{P}(H_1) \times \dotsm \times \mathcal{P}(H_t)$; and
        \item the generating function of the pseudocodewords of $C(H)$, $\mathcal{P}(H)$, is given by
        \[ f_H(\mathbf{x}_1, \dotsc, \mathbf{x}_t) = \prod_{k=1}^t f_{H_k}(\mathbf{x}_k).\]
    \end{enumerate}
\end{thm}

\begin{proof}
Assume $H_i \in \F_2^{r_i \times n_i}$ so that $r = \sum_{i=1}^t r_i$ and $n = \sum_{i=1}^t n_i$. We define $L_k$ to be the $k$th block row of $H$. Set $n_k^{(0)} = \sum_{\ell=1}^{k-1} n_\ell$ and $n_k^{(1)} = \sum_{\ell=k+1}^t n_\ell$. Then we can write $ L_k = \begin{bmatrix} 0_{r_k \times n_k^{(0)}} & H_k & 0_{r_k \times n_k^{(1)}} \end{bmatrix}$.
Now, $$R(H)=\bigcap_{k\in [t]}\bigcap_{j\in [r_k]} \poly(\row_j (L_k)).$$
Since $\bigcap_{j\in [r_k]} \poly(\row_j (L_k))=\mathbf{0}_{n_k^{(0)}}\times R(H_k) \times\mathbf{0}_{n_k^{(1)}}$, 1) readily follows. We can also see that
\begin{align*}
    \mathcal{K}(L_k)&= \left\{ \mathbf{v} \in \R_{\geq 0}^n : \row_j(L_k) \cdot \mathbf{v} \geq 2(L_k)_{ji} v_i \ \forall i \in [n], \forall j \in [r_k]\right\} \\
    &= \left\lbrace \mathbf{v} \in \R_{\geq 0}^n : 
    \begin{array}{l}
        \row_j(0_{r_i \times n_k^{(0)}}) \cdot (v_1, \dotsc, v_{n_k^{(0)}}) + \\\row_j(H_k) \cdot (v_{n_k^{(0)}+1}, \dots, v_{n_k^{(0)} + n_k}) \\
    +  \row_j(0_{r_i \times n_k^{(1)}}) \cdot (v_{n_k^{(0)}+n_k+1}, \dotsc, v_n) \geq 2(L_k)_{ji} v_i \\
    \forall i \in [n], \forall j \in [r_k]
    \end{array}  \right\rbrace
\end{align*}

Because $\row_j(0_{r_i \times n_k^{(0)}}) \cdot (v_1, \dotsc, v_{n_k^{(0)}}) = 0$ and $\row_j(0_{r_i \times n_k^{(1)}}) \cdot (v_{n_k^{(0)}+n_k+1}, \dotsc, v_n) = 0$ for all $j$ and for any choices of $(v_1, \dotsc, v_{n_k^{(0)}}) \in \R_{\geq 0}^{n_k^{(0)}}$ and $(v_{n_k^{(0)}+n_k+1}, \dotsc, v_n) \in \R_{\geq 0}^{n_k^{(1)}}$, we see that
\begin{align*}
    \mathcal{K}(L_k) &= \R_{\geq 0}^{n_k^{(0)}} \times \mathcal{K}(H_k) \times \R_{\geq 0}^{n_k^{(1)}} \\
    &= \R_{\geq 0}^{n_1} \times \dotsm \times \R_{\geq 0}^{n_{k-1}} \times \mathcal{K}(H_k) \times \R_{\geq 0}^{n_{k+1}} \times \dotsm \times \R_{\geq 0}^{n_t}.
\end{align*}
By Lemma~\ref{lem:intersection}, 2) holds since
\[\mathcal{K}(H) = \bigcap_{k=1}^t \mathcal{K}(L_k) = \mathcal{K}(H_1) \times \dotsm \times \mathcal{K}(H_t).\]

Let $\mathbf{p} \in \mathcal{P}(H)$. Then using 2) and  Theorem~\ref{thm:equiv}, $\mathbf{p} \in \mathcal{K}(H) = \mathcal{K}(H_1) \times \dotsm \times \mathcal{K}(H_t)$ and $H\mathbf{p}^T = \mathbf{0} \mod 2$. We can write $\mathbf{p} = (\mathbf{p}_1, \dotsc, \mathbf{p}_t)$ where $\mathbf{p}_k \in \R_{\geq 0}^{r_k}$. From this decomposition, it is clear that $\mathbf{p}_k \in \mathcal{K}(H_k)$. Because $H \mathbf{p}^T = \mathbf{0} \mod 2$,  we conclude that $H_k \mathbf{p}_k^T = \mathbf{0} \mod 2$. Hence $\mathbf{p} \in \mathcal{P}(H_1) \times \dotsm \times \mathcal{P}(H_t)$. To see that $ \mathcal{P}(H_1) \times \dotsm \times \mathcal{P}(H_t) \subseteq \mathcal{P}(H)$, let $\mathbf{p} = (\mathbf{p}_1, \dotsc, \mathbf{p}_t) \in \mathcal{P}(H_1) \times \dotsm \times \mathcal{P}(H_t)$. Again using (1) and Theorem~\ref{thm:equiv}, $\mathbf{p} \in \mathcal{K}(H) = \mathcal{K}(H_1) \times \dotsm \times \mathcal{K}(H_t)$, $\mathbf{p}_k \in \mathcal{K}(H_k)$, and $H_k \mathbf{p}_k^T = \mathbf{0} \mod 2$ for all $k \in [t]$. Thus, $H \mathbf{p}^T = \sum_{k=1}^t H_k \mathbf{p}_k^T = \mathbf{0} \mod 2$, and so $\mathbf{p} \in \mathcal{P}(H)$. Hence $\mathcal{P}(H) = \mathcal{P}(H_1) \times \dotsm \times \mathcal{P}(H_t)$, proving 3).

To prove 4), let $f_{H_k}(\mathbf{x}_k)$ be the generating function of the pseudocodewords of $C(H_k)$. Then

\[ f_{H_k}(\mathbf{x}_k) = \sum_{\mathbf{p}_k \in \mathcal{P}(H_k)} \mathbf{x}_k^{\mathbf{p}_k} = \sum_{\mathbf{p}_k \in \mathcal{P}(H_k)} x_{k1}^{p_{k1}} \dotsm x_{k n_k}^{p_{k n_k}}.\]
By 2), $\mathcal{P}(H) = \mathcal{P}(H_1) \times \dotsm \times \mathcal{P}(H_t)$. Therefore, 4) holds as
\begin{align*}
    f_H(\mathbf{x}_1, \dotsc, \mathbf{x}_t) &= \sum_{(\mathbf{p}_1, \dotsc, \mathbf{p}_t) \in \mathcal{P}(H_1) \times \dotsm \times \mathcal{P}(H_t)} \mathbf{x}_1^{\mathbf{p}_1} \dotsm \mathbf{x}_t^{\mathbf{p}_t} \\
    &= \left( \sum_{\mathbf{p}_1 \in \mathcal{P}(H_1)} \mathbf{x}_1^{\mathbf{p}_1} \right) \dotsm \left( \sum_{\mathbf{p}_t \in \mathcal{P}(H_t)} \mathbf{x}_t^{\mathbf{p}_t} \right) \\
    &= \prod_{k=1}^t f_{H_k}(\mathbf{x}_k).
\end{align*}
\end{proof}

\begin{thm}\label{thm:blockrow}
    Let $H = \begin{bmatrix} H_1 & H_2 & \dotsm & H_t\end{bmatrix} \in \F_2^{r \times n}$.
    Then
    \begin{enumerate}
        \item $\mathcal{K}(H_1) \times \dotsm \times \mathcal{K}(H_t) \subseteq \mathcal{K}(H)$;
        \item $\mathcal{P}(H_1) \times \dotsm \times \mathcal{P}(H_t) \subseteq \mathcal{P}(H)$; and
        \item $f_H(\mathbf{0}, \dotsc, \mathbf{0}, \mathbf{x}_i, \mathbf{0}, \dotsc, \mathbf{0}) = f_{H_i}(\mathbf{x}_i)$.
    \end{enumerate}
\end{thm}

\begin{proof}
Suppose $H_k \in \F_2^{r \times n_k}$ so that $n = \sum_{k=1}^t n_k$. Let $\mathbf{w}=(\mathbf{v}_1,\ldots,\mathbf{v}_t)\in\mathcal{K}(H_1) \times \dotsm \times \mathcal{K}(H_t)$.
    Because $\mathbf{v}_k \in \mathcal{K}(H_k)$ for all $k \in [t]$, for all $j \in [r]$ and for all $i \in [n_k]$, 
    \[ \text{Row}_j (H_k) \cdot \mathbf{v}_k \geq 2 (H_k)_{ji} v_{ki}.\]
    For all $i \in [n]$, $j \in [r]$, and for any $k \in [t]$,
    \begin{align*}
        \text{Row}_j (H) \cdot \mathbf{w} &= \text{Row}_j (H_1) \cdot \mathbf{v}_1 + \dots + \text{Row}_j (H_t) \cdot \mathbf{v}_t \\
        &\geq 2 (H_k)_{ji} v_{ki}.
    \end{align*}
    Hence, $\mathbf{w} \in \mathcal{K}(H)$. It is an immediate consequence that $\mathcal{K}(H_1) \times \dotsm \times \mathcal{K}(H_t) \subseteq \mathcal{K}(H)$, proving 1).

    Let $\mathbf{p} = (\mathbf{p}_1, \dotsc, \mathbf{p}_t) \in \mathcal{P}(H_1) \times \dotsm \times \mathcal{P}(H_t)$. Using 1) and Theorem~\ref{thm:equiv}, $\mathbf{p}_i \in \mathcal{K}(H_i)$ and $H_i\mathbf{p}_i^T = \mathbf{0} \pmod 2$ for all $i \in [t]$. Thus, $H \mathbf{p}^T = \sum_{k=1}^t H_k \mathbf{p}_k^T = \mathbf{0} \pmod 2$ and so $\mathbf{p} \in \mathcal{P}(H)$, proving 2).
    
    To prove 3), let $i \in [t]$. It is clear from 2) that $f_H(\mathbf{0}, \dotsc, \mathbf{0}, \mathbf{x}_i, \mathbf{0}, \dotsc, \mathbf{0})$ contains all the terms from $f_{H_i}(\mathbf{x}_i)$. Let $\mathbf{x}_i^{\mathbf{p}_i}$ be a term of $f_{H_i}(\mathbf{x}_i)$. Because $(\mathbf{0}, \dotsc, \mathbf{0}, \mathbf{p}_i, \mathbf{0}, \dotsc \mathbf{0}) \in \mathcal{P}(H)$, $\mathbf{x}_i^{\mathbf{p}_i}$ is also a term of $f_H(\mathbf{x})$. In particular, $\mathbf{x}_i^{p_i}$ has no entries from $\mathbf{x}_1, \dotsc, \mathbf{x}_{i-1}, \mathbf{x}_{i+1}, \dotsc, \mathbf{x}_t$, and so $\mathbf{x}_i^{\mathbf{p}_i}$ is a term of $f_H(\mathbf{0}, \dotsc, \mathbf{0}, \mathbf{x}_i, \mathbf{0}, \dotsc, \mathbf{0})$, proving 3).
\end{proof}

When the matrices $H_i$ in Theorem~\ref{thm:blockrow} are all the same matrix $H$, we can provide an additional characterization of elements of the fundamental cone based on the elements of the fundamental cone of $H$.

\begin{thm} \label{thm:h-to-hhhh}
    Let $H \in \F_2^{r \times n}$, $\mathbf{v} \in \mathcal{K}(H)$, $H' = \begin{bmatrix} H & \dotsm & H \end{bmatrix} \in \F_2^{r \times tn}$, and
    \[\mathbf{w} = (w_{11}, w_{12}, \dotsc, w_{1n}, w_{21}, \dotsc, w_{2n}, \dotsc, w_{t1}, \dotsc, w_{tn}) \in \R^{tn}_{\geq 0}.\]
    If $w_{ki} \leq v_i \leq \sum_{j=1}^t w_{ji}$ for all $i \in [n]$ and $k \in [t]$, then $\mathbf{w} \in \mathcal{K}(H')$.
\end{thm}

\begin{proof}
    Because $\mathbf{v} \in \mathcal{K}(H)$, for all $j \in [r]$ and all $i \in [n]$, 
    \[ \text{Row}_j (H) \cdot \mathbf{v} \geq 2 h_{ji} v_i. \]
    We can rewrite:
    \begin{align*}
        \text{Row}_j (H) \cdot \mathbf{v} &\leq \text{Row}_j (H) \cdot \left( \sum_{k=1}^{t} w_{k1}, \dotsc, \sum_{k=1}^{t} w_{kn} \right) \\
        &= \text{Row}_j(H) \cdot (w_{11}, \dotsc, w_{1n}) + \dotsm + \text{Row}_j(H) \cdot (w_{t1}, \dotsc, w_{tn}) \\
        &= \text{Row}_j(H') \cdot \mathbf{w}
    \end{align*}
    Putting everything together,
    \[ \text{Row}_j(H') \cdot \mathbf{w} \geq \text{Row}_j(H) \cdot \mathbf{v} \geq 2 h_{ji} v_i \geq 2 h_{ji} w_{ki}\]
    for all $i \in [n]$, $j \in [r]$, and $k \in [t]$. Hence $\mathbf{w} \in \mathcal{K}(H')$.
\end{proof}

The next result provides a restriction on the fundamental cone of the code from a parity-check matrix which has been augmented by a single additional column.

\begin{pr} \label{pr:column}
    Let $H_1 \in \F_2^{r \times n}$, $\mathbf{s} \in \F_2^r$, $\supp(\mathbf{s}) = S \subseteq [r]$, and $H = \begin{bmatrix} H_1 & \mathbf{s}^T \end{bmatrix}$. Then
    \[   \left\lbrace (\mathbf{v},w) \in \mathcal{K}(H_1) \times \R_{\geq 0} : w \leq \row_j(H_1) \cdot \mathbf{v} \ \forall j \in S \right\rbrace \subseteq \mathcal{K}(H).\]
Moreover, for $H = \begin{bmatrix} H_1 & J \end{bmatrix}$ where $J \in \F_2^{r \times r}$ is a permutation matrix given by $\sigma \in S_r$ such that
\[J_{ji} = \begin{cases} 1 & \text{if }\sigma(j) = i, \\ 0 & \text{otherwise,} \end{cases}\]
we have
    \[\hspace*{-2pt} \left\lbrace (\mathbf{v}, \mathbf{w}) \in \mathcal{K}(H_1) \times \R_{\geq 0}^r \!:\! w_{\sigma(j)} \leq \row_j (H_1) \cdot \mathbf{v} \ \forall j \in [r] \right\rbrace \subseteq  \mathcal{K}(H).\]
  \end{pr}

\begin{proof}
    Let $(\mathbf{v},w) \in \mathcal{K}(H_1) \times \R_{\geq 0}$ where $w \leq \row_j(H_1) \cdot \mathbf{v}$ for all $j \in S$.
    For all $j \in [r] \setminus S$ and $i \in [n]$,
    \[ \row_j(H) \cdot (\mathbf{v},w) = \row_j(H_1) \cdot \mathbf{v} + 0 \cdot w \geq 2 (H_1)_{ji}v_i = 2(H)_{ji} v_i, \]
    and for $i = n+1$,
    \[ \row_j(H) \cdot (\mathbf{v},w) = \row_j(H_1) \cdot \mathbf{v} + 0 \cdot w \geq 0 = 2(H)_{ji} w.\]
    Similarly, for all $j \in S$ and $i \in [n]$, 
    \[ 
    \row_j(H) \cdot (\mathbf{v},w) \geq \row_j(H_1) \cdot \mathbf{v} \geq 2 (H_1)_{ji} v_i + 1 \cdot w \geq 2(H)_{ji} v_i, \]
    and for $i =n+1$,
    \[ \row_j(H) \cdot (\mathbf{v},w) = \row_j(H_1) \cdot \mathbf{v} + 1 \cdot w \geq 2 w = 2(H)_{ji} w, \]
    showing that $(\mathbf{v},w) \in \mathcal{K}(H)$.

    The result for $H = \begin{bmatrix} H_1 & J \end{bmatrix}$ follows from applying the previous argument repeatedly, since each column and row of $J$ has weight $1$.
\end{proof}

In general, the containments in Proposition~\ref{pr:column} are proper, as the next example illustrates.

\begin{exmp} \label{ex:hammingpsc}
    Consider the Hamming code $C(H)$ with
    \[ H = \begin{bmatrix} 1 & 0 & 1 & 1 & 1 & 0 & 0 \\
                           0 & 1 & 0 & 1 & 1 & 1 & 0 \\
                           0 & 0 & 1 & 0 & 1 & 1 & 1 \end{bmatrix}.
                           \]
    Then vector $\mathbf{v} = (2,0,0,1,1,0,1) \in \mathcal{K}(H)$. However, consider  $\mathbf{w}' = (\mathbf{w},w_8) = (2,0,0,2,1,0,1,2)$. For any augmented matrix $H' = \begin{bmatrix} H & \mathbf{s}^T \end{bmatrix}$ where $\mathbf{s} \in \F_2^3$ and $\supp(\mathbf{s}) \supseteq \{2\}$, $\mathbf{w}' \in \mathcal{K}(H')$. However, $\mathbf{w} \not\in \mathcal{K}(H)$.
\end{exmp}

We close this section by noting that there are other, arguably more sophisticated, ways to determine pseudocodewords or useful information from a code. In particular, the classical weight enumerator is a specialized generating function that captures the distribution of codeword weights. For cycle codes, there is a combinatorial connection between the generating function of the pseudocodewords and the edge zeta function of the Tanner graph \cite{KLVW}. According to \cite{BP}, $f_H({\bf x})$ has a representation as a short rational function, meaning that it may be compactly represented despite having many terms. This allows a simple yet precise enumeration for the pseudocodewords. Here, Lemma~\ref{lem:intersection} can be used to derive $f_H({\bf x})$ from a collection of generating functions in rational form.

\begin{pr}
 Consider
    \[H = \begin{bmatrix} H_1 \\ H_2 \\ \vdots \\ H_t \end{bmatrix} \in \F_2^{r \times n}\]
 and fix $s$, the maximum number of binomials in each partial fraction of $f_{H_i}({\bf x})$. Then there exists a polynomial time algorithm that computes the generating function $f_H({\bf x})$ of the pseudocodewords of $H$, given the generating functions $f_{H_i}({\bf x})$ for each $\mathcal{P}(H_i)$, $i \in [t]$. 
\end{pr}

\begin{proof}
   This follows from Lemma~\ref{lem:intersection} and \cite[Theorem 3.6]{BW}.
\end{proof}

\section{Properties of the fundamental cone for specific classes of codes} \label{sec:classes}

In this section, we will apply the general theory we developed in Section~\ref{sec:fund_cone} to three specific classes of LDPC codes: quantum stabilizer codes in Section~\ref{sec:quantum}, quasi-cyclic codes in Section~\ref{sec:qc}, and spatially-coupled codes in Section~\ref{sec:sc}.

\subsection{Pseudocodewords of quantum stabilizer codes} \label{sec:quantum}

This subsection explores the structure of pseudocodewords from quantum stabilizer codes, with a focus on CSS codes, obtained from the celebrated Calderbank-Shor-Steane construction \cite{CS96,S96}. We will see that through the use of label codes, which are binary vector representations of a stabilizer and its normalizer, we may apply the framework from Section \ref{sec:fund_cone} to stabilizer and normalizer codes \cite{CRSS98,Nielsen_Chuang_2010}.

We begin by outlining the necessary concepts to understand quantum stabilizer codes and refer readers to \cite{marinescu2011classical} for a review on quantum information. Consider the Pauli matrices 
\[ X = \begin{bmatrix} 0 & 1 \\ 1 & 0 \end{bmatrix}, \quad Y = \begin{bmatrix} 0 & -i \\ i & 0 \end{bmatrix}, \quad Z = \begin{bmatrix} 1 & 0 \\ 0 & -1 \end{bmatrix} \]
and the $n$-qubit Pauli group
$$\mathcal G_n: = \left\{ c \mathcal P_1 \otimes \dots \otimes \mathcal P_n : c \in \left\{ \pm 1, \pm i \right\}, \mathcal P_i \in \left\{ X, Y, Z, I \right\} \right\},$$
where $I$ is the $2 \times 2$ identity matrix. Instead of working directly with stabilizer quantum error correcting codes, we consider a code's stabilizer label code and normalizer label code. As in \cite{Li_Vontobel_ISIT}, we make use of the map $\ell: \mathcal G_1 \rightarrow \F_2^2$ given by 
\begin{align*}
    X &\mapsto (1,0) & Z &\mapsto (0,1) \\
    Y &\mapsto (1,1) & I &\mapsto (0,0).
\end{align*}
For $i \in [2]$, let $\pi_i:\F_2^n \rightarrow  \F_2$ be the projection map given by $\mathbf{x}  \mapsto  x_i$. Further, define $\phi_X$ and $\phi_Z$ as follows:
$$
\begin{array}{lccc}
\phi_X: &\mathcal G_n &\rightarrow &\F_2^{n}\\
& \mathcal P_1 \otimes \dots \otimes \mathcal P_n & \mapsto & \left( \pi_1 \circ \ell(\mathcal P_1), \dots , \pi_1 \circ \ell(\mathcal P_n) \right)
\end{array}
$$
and 
$$
\begin{array}{lccc}
\phi_Z: &\mathcal G_n &\rightarrow &\F_2^{n}\\
& \mathcal P_1 \otimes \dots \otimes \mathcal P_n & \mapsto & \left( \pi_2 \circ \ell(\mathcal P_1), \dots , \pi_2 \circ \ell(\mathcal P_n) \right).
\end{array}
$$

Let $S:= \langle g_1, \dotsc, g_r\rangle \leq \mathcal G_n$. Associated with $S$ is the binary label code $A$ \cite[Definition 2]{Li_Vontobel_ISIT}, which is a binary code with generator matrix 
$$
A(S):=\left[ 
\begin{array}{cc}
\phi_X (g_1) &\phi_Z(g_1) \\
\phi_X (g_2) &\phi_Z(g_2) \\
\vdots & \vdots \\
\phi_X (g_r) &\phi_Z(g_r) \end{array}
\right] \in \F_2^{r \times 2n}.
$$
We typically assume no redundancy so that $r = n-k$. For convenience, we will sometimes write $A(S):=[\phi_X(A) \mid  \phi_Z(A)]$. The binary normalizer label code $\mathcal{N}$ \cite[Definition 4]{Li_Vontobel_ISIT} is the dual code of $A$ under the symplectic inner product. Equivalently, $\mathcal{N}$ is the code where $A(S)$ is the parity-check matrix, i.e. 
$$\mathcal N(S):=C \left([\phi_X(A) \mid  \phi_Z(A)] \right) \in \F_2^{2n}.$$
This correspondence will be useful in our analysis of quantum stabilizer codes. 

The codeword polytope of $\mathcal N(S)$ is the convex hull of the codewords of the associated binary normalizer label code: $$\conv(\mathcal N (S) ) = \poly(\mathcal N(S)) \subseteq \R^{2n}.$$ 
As discussed in Section~\ref{sec:preliminaries}, this polytope depends only on the set of codewords and not the representation of the parity-check matrix. While this polytope gives the precise blockwise MAP decoding in the setting considered in \cite{Li_Vontobel_ISIT}, finding its precise description is often challenging. The tools developed in Section \ref{sec:fund_cone} for the fundamental cone and relaxed polytope of traditional binary codes apply immediately to quantum normalizer label codes. 

We can now apply these results to stabilizer codes by using the normalizer. The following result follows immediately from the definition of $A(S)$ and Lemma~\ref{lem:intersection}.

\begin{thm} \label{thm:stabilizer}
The fundamental cone of the normalizer code $C$ given by 
 $S:=\left< g_1, \dots, g_r \right> \leq \mathcal G_n$ is 
    \[  \bigcap_{i=1}^r \mathcal{K}( \left[ \phi_X(g_i) \ \phi_Z(g_i) \right]).\]
\end{thm}

The well-known Calderbank-Shor-Steane (CSS) code can be constructed as a stabilizer quantum error-correcting code whose label code has parity check matrix of the form 
\begin{equation} H = \begin{bmatrix} H_1 & 0 \\ 0 & H_2\end{bmatrix} \label{eq:css} \end{equation}
where $H_1, H_2 \in \F_2^{r \times n}$ are such that $H_1 H_2^T = {0}_{r \times r}$, meaning that $C(H_2)^\perp \subseteq C(H_1)$. By applying Theorem~\ref{thm:products} to Eq.~\ref{eq:css}, one obtains the following result on pseudocodewords of CSS codes.

\begin{thm} \label{pr:css}
    Let $C$ be a CSS code with parity-check matrix as given in Equation~\ref{eq:css}. Then the pseudocodewords of $C$ are given by $\mathcal{P}(H_1) \times \mathcal{P}(H_2)$ and the generating function of the pseudocodewords of $C$ is
    \[ f_H(\mathbf{x}) = \left(\sum_{\mathbf{p} \in \mathcal{P}(H_1)} x_1^{p_1} \dotsm x_n^{p_n} \right) \left( \sum_{\mathbf{p}' \in \mathcal{P}(H_2)} x_{n+1}^{p_1'} \dotsm x_{2n}^{p_n'} \right).\]
\end{thm}

Since Theorem~\ref{pr:css} allows us to handle submatrices of $H$ separately, as a result one can greatly simplify the process of finding the pseudocodewords of CSS codes. 

One well-known family of CSS codes are the Steane codes, which are CSS codes for which $H_1 = H_2 = H_r$, where $H_r$ is the parity-check matrix of the $[2^r-1, 2^r-r-1,3]$ Hamming code with $r \geq 3$. We illustrate the usefulness of Theorem~\ref{pr:css} with the first code in this family, where $r = 3$.

\begin{exmp} \label{ex:steane} The first code in the family of Steane codes is given by $\widetilde{H} = \begin{bmatrix} H & 0 \\ 0 & H\end{bmatrix}$ where $H$ is the parity-check matrix of the classical $[7,4,3]$ binary Hamming code $C(H)$. With the parity-check matrix representation $H$ as in Example~\ref{ex:hammingpsc}, Polymake \cite{polymake} computes there to be 42 edges of the fundamental cone and enumerates 96 vertices of the relaxed polytope of $C(H)$, which are exactly the LP pseudocodewords of $C(H)$. By Theorem~\ref{thm:products}, there are exactly $96^2 = 9216$ LP pseudocodewords of $C(\widetilde{H})$, and they are of the form $\mathbf{v} = (\mathbf{v}_1, \mathbf{v}_2)$ where $\mathbf{v}_1, \mathbf{v}_2 \in R(H)$.
\end{exmp}

\subsection{Pseudocodewords of quasi-cyclic codes} \label{sec:qc}

In this subsection, we first provide general results on pseudocodewords of classical quasi-cyclic codes. At the end of the section, we additionally use results from Section~\ref{sec:quantum} and apply these results to quasi-cyclic CSS codes.

A linear code $C$ is called a \textit{cyclic code} if every cyclic shift of a codeword in $C$ is also a codeword in $C$. Given $\mathbf{x} = (x_1, \dotsc, x_n)$,  
$ \mathbf{x}^i = (x_{n-i+1}, x_{n-i}, \dotsc, x_{n}, x_1, \dotsc, x_{n-i})$ is its cyclic shift $i$ positions to the right. A linear code is \textit{quasi-cyclic with shifting constraint $n_0$}  if  cyclically shifting a codeword a fixed number $n_0$ of symbol positions results in another codeword. Cyclic codes are quasi-cyclic codes with shifting constraint $n_0 = 1$.

Given shifting constraint $n_0$ and positive integers $c, t$ with $c < t$, let $H_i \in \F_q^{c \times n_0}$ for $i \in [t]$. The nullspace of the matrix $H \in \F_q^{tc \times t n_0}$ in Equation~\ref{eq:blockcirculant}, which is said to be in block circulant form, is a quasi-cyclic parity-check code.

\noindent\begin{minipage}{.47\linewidth}
\begin{equation}
   H = \begin{bmatrix} H_1 & H_2 & \dotsm & H_t \\
                       H_t & H_1 & \dotsm & H_{t-1} \\
                       \vdots & \vdots & \ddots & \vdots \\
                       H_2 & H_3 & \dotsm & H_1 \end{bmatrix}
                       \label{eq:blockcirculant}
\end{equation}
\end{minipage}%
\begin{minipage}{.52\linewidth}
\begin{equation}
   H = \begin{bmatrix} J_{1,1} & J_{1,2} & \dotsm & J_{1,n_0} \\
                         J_{2,1} & J_{2,2} & \dotsm & J_{2,n_0} \\
                         \vdots & \vdots & \ddots & \vdots \\
                         J_{c,1} & J_{c,2} & \dotsm & J_{c,n_0} \end{bmatrix}
                         \label{eq:circulant}
\end{equation}
\end{minipage}

Such a matrix is often represented as a matrix consisting of circulant matrices which is permutation equivalent to $H$. In this case, $H$ is a $c \times n_0$ block matrix composed of $t \times t$ circulants, as in Equation~\ref{eq:circulant}.

The next corollary follows immediately from Theorem~\ref{thm:blockrow} as well as Lemma~\ref{lem:intersection}.

\begin{co} \label{co:qc}
Let $H$ be the parity-check matrix of a quasi-cyclic parity-check code in circulant form, as in Equation~\ref{eq:circulant}. Then 
\[\bigcap_{j=1}^c \left(\mathcal{K}(J_{j,1}) \times \dotsm \times \mathcal{K}(J_{j,n_0})\right) = \prod_{i=1}^{n_0} \left( \bigcap_{j=1}^c \mathcal{K}(J_{j,i}) \right) \subseteq \mathcal{K}(H).\]
\end{co}

We can make use of the highly structured nature of cyclic and quasi-cyclic parity-check codes in decoding algorithms.

\begin{defi}
    A parity-check matrix $H$ (resp., polytope $\mathscr{P}$) is called quasi-cyclic with shifting constraint $n_0$ if every quasi-cyclic shift by $n_0$ of a row of $H$ (resp., vertex of $\mathscr{P}$) is a row of $H$ (resp., vertex of $\mathscr{P}$).
\end{defi}

As mentioned in \cite{Heidarpour}, given a cyclic code with a parity-check matrix that does not consist of all possible cyclic shifts of its rows, all cyclic shifts of integer vertices of the fundamental polytope are also vertices. However, this is not necessarily the case for non-integer vertices. In \cite{Heidarpour}, it is shown that if a parity-check matrix of a cyclic code contains all possible cyclic shifts of its rows, non-integer vertices are also vertices of the fundamental polytope. We prove analogous results for quasi-cyclic parity-check codes.

\begin{thm} \label{thm:qcpolytope}
    A polytope $\poly(H)$ is quasi-cyclic with shifting constraint $n_0$ if the matrix $H$ is quasi-cyclic with shifting constraint $n_0$.
\end{thm}

\begin{proof}
    Let $H$ be a quasi-cyclic parity-check matrix with shifting constraint $n_0$. It is clear that the set of inequalities given by $\mathcal{K}(H)$ in Definition~\ref{def:cone} is both completely determined by and completely determines $H$. Let $H^{n_0 i}$ be $H$ with its rows all cyclically shifted $n_0 i$ positions to the right. Because $H$ is quasi-cyclic, $H^{n_0 i}$ for $i \in \Z$ is the same as $H$ up to row permutations, and so the set of inequalities in $\mathcal{K}(H)$ is the same as those in $\mathcal{K}(H^{n_0 i})$, i.e. $\mathcal{K}(H) = \mathcal{K}(H^{n_0 i})$.

    Let $\mathbf{x} \in \mathcal{K}(H)$. Note that $H (\mathbf{x}^{n_0 i})^T = H^{n_0 i} \mathbf{x}^T$, so $\mathbf{x}^{n_0 i} \in \mathcal{K}(H)$. The fundamental cone is quasi-cyclic if $H$ and $H^{n_0 i}$ have the same fundamental cone, which is satisfied if $H$ contains all quasi-cyclic shifts of its rows. Using Theorem~\ref{thm:equiv}, we conclude this is also true of the polytopes of $H$ and $H^{n_0 i}$.
\end{proof}

A crucial property of LP decoder performance on a code with quasi-cyclic parity-check matrix is given in the following theorem.

\begin{thm} \label{thm:qcdecoder}
    If the parity-check matrix is quasi-cyclic with shifting constraint $n_0$, then over the binary symmetric channel, a shift of the error vector by $n_0 i$ for any $i \in \Z$ does not affect the success or failure of the LP decoder.
\end{thm}

\begin{proof}
    The probability that the LP decoder fails is independent of the codeword which was transmitted \cite{FWK}, so we can assume that the transmitted codeword is the all zero vector. Suppose $\mathbf{e}$ is the error vector, and hence $\mathbf{e}$ is also our received vector. If the cost of vertex $\mathbf{v}$ is minimum, then because of the quasi-cyclic property of the polytope, the cost of the vertex $\mathbf{v}^{n_0 i}$ is also minimum for the error pattern $\mathbf{e}^{n_0 i}$. Hence, when error vectors are quasi-cyclic shifts of each other, the corresponding LP decoder outputs are also quasi-cyclic shifts of each other. This means that if the output of the LP decoder for some vector $\mathbf{v}$ is integer (or non-integer), then the LP decoder outputs for $\mathbf{v}^{n_0 i}$ are also integer (or non-integer).
\end{proof}

The performance of an LP decoder is greatly dependent on the parity-check matrix representation of the code. As shown in e.g. \cite{FWK,KS07}, introducing redundancy to a parity-check matrix by adding additional rows can improve LP decoder performance. However, adding these redundant rows randomly is unlikely to have significant positive impact. Theorems~\ref{thm:qcpolytope} and \ref{thm:qcdecoder} suggest an algorithmic approach to finding ``good'' redundant rows to improve LP decoder performance. The algorithm, based on that for cyclic codes in \cite{Heidarpour}, is as follows:

\begin{enumerate}
    \item Identify a low weight codeword among the codewords of the dual code that are not rows of the parity-check matrix, $H$.
    \item Add this codeword and all its quasi-cyclic shifts to $H$.
    \item Examine the LP decoder performance. If it achieves the required performance, stop. Otherwise, return to 1).
\end{enumerate}

We conclude this section by providing an additional application of this theory for quasi-cyclic codes to quantum stabilizer codes.

CSS codes using cyclic and quasi-cyclic blocks were first introduced in \cite{li2004quantumcyclic} and \cite{hagiwara2007quantum}, respectively. These are CSS codes for which $H_1$ and $H_2$ in Equation~\ref{eq:css} are cyclic (resp. quasi-cyclic) parity-check matrices. Notably, because all binary Hamming codes have cyclic representations, the family of Steane codes have cyclic CSS code representations.

The following corollary captures the result for the CSS construction given by Theorem~\ref{pr:css}.

\begin{co} \label{co:qccss}
    Let $H_1 = (J_{j,i })$ for $j \in [c_1]$ and $i \in [n_1]$ and $H_2 = (K_{j,i})$ for $j \in [c_2]$ and $i \in [n_2]$ be two quasi-cyclic matrices composed of block circulants and let $C$ be a CSS code from the parity-check matrix $H$ as given in Equation~\ref{eq:css}. Then 
    \[ \prod_{i=1}^{n_0} \left( \bigcap_{j=1}^{c_0} \mathcal{K}(J_{j,i}) \right) \times \prod_{i=1}^{n_1} \left( \bigcap_{j=1}^{c_1} \mathcal{K}(K_{j,i}) \right) \subseteq \mathcal{K}(H_1) \times \mathcal{K}(H_2) = \mathcal{K}(H).\]
\end{co}

\begin{proof}
    This follows immediately from Theorem~\ref{pr:css} and Corollary~\ref{co:qc}.
\end{proof}

\begin{exmp}
The Hamming code $C(H)$ given in Example \ref{ex:hammingpsc} is cyclic, and so $C(\widetilde{H})$ where $\widetilde{H} = \begin{bmatrix} H & 0 \\ 0 & H\end{bmatrix}$ as in Example~\ref{ex:steane},  is the first Steane code, and is hence a quasi-cyclic CSS code. By adding the remaining four cyclic shifts of the first row to $H$, we obtain a representation with no non-codeword pseudocodewords. $\widetilde{H}$ thus also has no non-codeword pseudocodewords. A similar approach can be taken for any larger Steane code or any quasi-cyclic CSS code.
\end{exmp}

We now give an example of a quasi-cyclic quantum LDPC code. We use the code given in Example 5.5 in \cite{hagiwara2007quantum}.

\begin{exmp}
    Let $\mathcal{H}_C$ and $\mathcal{H}_D$ be given as below.
    \begin{equation*}
        \mathcal{H}_C = \begin{bmatrix} 1 & 2 & 4 & 3 & 6 & 5 \\
        4 & 1 & 2 & 5 & 3 & 6 \\
        2 & 4 & 1& 6 & 5 & 3 \end{bmatrix},
        \quad
        \mathcal{H}_D = \begin{bmatrix} 4 & 2 & 1 & 6 & 3 & 5 \\
        1 & 4 & 2 & 5 & 6 & 3 \\
        2 & 1 & 4 & 3 & 5 & 6\end{bmatrix}.
    \end{equation*}
    Replace each entry $c_{i,j}$ of $\mathcal{H}_C$ and $d_{i,j}$ of $\mathcal{H}_D$ with a $7 \times 7$ binary circulant permutation matrix $P_{c_{i,j}}$ or $P_{d_{i,j}}$ that is obtained from the identity matrix by cyclically shifting the rows  right by $c_{i,j}$ or $d_{i,j}$ positions, respectively. Call the resulting $21 \times 42$ binary matrices $H_C$ and $H_D$. Let $H=\begin{bmatrix} H_C & 0 \\ 0 & H_D\end{bmatrix}$ be the parity-check matrix of the CSS code obtained by using the matrices $H_C$ and $H_D$. Then $C(H)$ is a binary quasi-cyclic quantum LDPC code.

    Corollary~\ref{co:qccss} can be applied to find some of the pseudocodewords of $C(H)$. It reduces the problem to determining the fundamental cone of the 
    $P_{c_{i,j}}$ and $P_{d_{i,j}}$. We see that 
     \[ \prod_{i=1}^{5} \left( \bigcap_{j=1}^{3} \mathcal{K}(P_{c_{j,i}}) \right) \times \prod_{i=1}^{5} \left( \bigcap_{j=1}^{3} \mathcal{K}(P_{c_{j,i})} \right) \subseteq \mathcal{K}(H_C) \times \mathcal{K}(H_D) = \mathcal{K}(H).\]
    Once found, if the LP decoder performance is not as desired, a low-weight word in the dual code and its quasi-cyclic shifts can be added to $H$.

    For example, the low-weight word 
    \[\mathbf{c} = (1001000, 0000000, 0000110, 0010010, 0010100, 0101000)\]
    is in $C(H_C)^\perp$ but not a row of $H_C$. Therefore, the rows
    
    \noindent\scalebox{.69}{
    \begin{minipage}{\linewidth}
    \[ \left[ \begin{array}{ccccccc|ccccccc|ccccccc|ccccccc|ccccccc|ccccccc} 
        1 & 0 & 0 & 1 & 0 & 0 & 0 & 0 & 0 & 0 & 0 & 0 & 0 & 0 & 0 & 0 & 0 & 0 & 1 & 1 & 0 & 0 & 0 & 1 & 0 & 0 & 1 & 0 & 0 & 0 & 1 & 0 & 1 & 0 & 0 & 0 & 1 & 0 & 1 & 0 & 0 & 0 \\
        0 & 1 & 0 & 0 & 1 & 0 & 0 & 0 & 0 & 0 & 0 & 0 & 0 & 0 & 0 & 0 & 0 & 0 & 0 & 1 & 1 & 0 & 0 & 0 & 1 & 0 & 0 & 1 & 0 & 0 & 0 & 1 & 0 & 1 & 0 & 0 & 0 & 1 & 0 & 1 & 0 & 0 \\
        0 & 0 & 1 & 0 & 0 & 1 & 0 & 0 & 0 & 0 & 0 & 0 & 0 & 0 & 1 & 0 & 0 & 0 & 0 & 0 & 1 & 1 & 0 & 0 & 0 & 1 & 0 & 0 & 0 & 0 & 0 & 0 & 1 & 0 & 1 & 0 & 0 & 0 & 1 & 0 & 1 & 0 \\
        0 & 0 & 0 & 1 & 0 & 0 & 1 & 0 & 0 & 0 & 0 & 0 & 0 & 0 & 1 & 1 & 0 & 0 & 0 & 0 & 0 & 0 & 1 & 0 & 0 & 0 & 1 & 0 & 1 & 0 & 0 & 0 & 0 & 1 & 0 & 0 & 0 & 0 & 0 & 1 & 0 & 1 \\
        1 & 0 & 0 & 0 & 1 & 0 & 0 & 0 & 0 & 0 & 0 & 0 & 0 & 0 & 0 & 1 & 1 & 0 & 0 & 0 & 0 & 0 & 0 & 1 & 0 & 0 & 0 & 1 & 0 & 1 & 0 & 0 & 0 & 0 & 1 & 1 & 0 & 0 & 0 & 0 & 1 & 0 \\
        0 & 1 & 0 & 0 & 0 & 1 & 0 & 0 & 0 & 0 & 0 & 0 & 0 & 0 & 0 & 0 & 1 & 1 & 0 & 0 & 0 & 1 & 0 & 0 & 1 & 0 & 0 & 0 & 1 & 0 & 1 & 0 & 0 & 0 & 0 & 0 & 1 & 0 & 0 & 0 & 0 & 1 \\
        0 & 0 & 1 & 0 & 0 & 0 & 1 & 0 & 0 & 0 & 0 & 0 & 0 & 0 & 0 & 0 & 0 & 1 & 1 & 0 & 0 & 0 & 1 & 0 & 0 & 1 & 0 & 0 & 0 & 1 & 0 & 1 & 0 & 0 & 0 & 1 & 0 & 1 & 0 & 0 & 0 & 0 \\
    \end{array}\right]\]
    \end{minipage}}
    
    \noindent can be added to $H_C$ to improve LP decoder performance. A similar process can be used to add rows to $H_D$.
\end{exmp}

\subsection{Pseudocodewords of spatially-coupled codes} \label{sec:sc}

We conclude this section with a relatively short subsection on spatially-coupled LDPC (SC-LDPC) codes, which are a special class of LDPC codes originally introduced in \cite{FZ99}. These codes have highly repetitive structures and exhibit excellent iterative decoder performance \cite{KRU11}.

SC-LDPC codes have parity-check matrices with $L$ block columns, each block column containing $m$ smaller parity-check matrices $H_i$ with $i \in [m]$, where each $H_i \in \F_2^{n_c \times n_r}$. There are two types of SC-LDPC codes: terminated and tailbiting. If $C$ is terminated, its parity-check matrix $H$ has form as in Eq.~\ref{eq:sc-terminated}, and if $C$ is tailbiting, $H$ has form as in Eq.~\ref{eq:sc-tailbiting}.

\noindent\begin{minipage}{.50\linewidth}
{\footnotesize\begin{equation}
   H = \begin{bmatrix} H_0 & & & & & & \\
                         H_1 & H_0 & & & & & \\
                         \vdots & & \ddots & & & &  \\
                         H_m & H_{m-1} &\dots & H_0 & & & \\
                         & H_m & \dots & H_1 & H_0 & & \\
                         & & \ddots & & & \ddots &  \\
                         & & & H_m & \dots & H_1 & H_0\\
                         & & & & \ddots & & \vdots  \\
                         & & & & & H_m &  H_{m-1}  \\
                         & & & & & & H_m \end{bmatrix} \label{eq:sc-terminated}
\end{equation}}
\end{minipage}%
\begin{minipage}{.49\linewidth}
{\footnotesize\begin{equation}
   H = \begin{bmatrix} H_0 & & & & H_m & \dots & H_1 \\
                                 H_1 & H_0 & & & & \ddots & \vdots \\
                                 \vdots & & \ddots & & & & H_m \\
                                 H_m & H_{m-1} &\dots & H_0 & & & \\
                                 & H_m & \dots & H_1 & H_0 & & \\
                                 & & \ddots & & & \ddots &  \\
                                 & & & H_m & \dots & H_1 & H_0 \end{bmatrix} \label{eq:sc-tailbiting}
\end{equation}}
\end{minipage}

\begin{co}
    Let $C$ be an SC-LDPC code with parity-check matrix $H$ as in either Eq.~\ref{eq:sc-terminated} or Eq.~\ref{eq:sc-tailbiting}. Then 
    \[ \underbrace{\left( \bigcap_{j=0}^m \mathcal{K}(H_j) \right) \times \dotsm \times \left( \bigcap_{j=0}^m \mathcal{K}(H_j) \right)}_{L \text{ times}} \subseteq \mathcal{K}(H).\]
\end{co}

\begin{proof}
    In both Eq.~\ref{eq:sc-terminated} and Eq.~\ref{eq:sc-tailbiting}, each $H_j$ appears exactly one time in each block column, and all other block matrices are zero matrices. So, by Lemma~\ref{lem:intersection}, the fundamental cone of a block column of $H$ is $\bigcap_{j=0}^m \mathcal{K}(H_j)$. By Theorem~\ref{thm:blockrow}, taking each block column of $H$ as an individual block matrix, the result follows.
\end{proof}

\section{Conclusion}
\label{sec:conclusions}
In this paper, we determined the pseudocodewords of codes defined using component codes. We applied these ideas to quantum stabilizer codes from the CSS construction, quasi-cyclic parity-check codes, and spatially-coupled LDPC codes. It remains to apply these results to the MDPC codes used in BIKE as well as other commonly considered families of quantum LDPC codes.

\bibliographystyle{vancouver}
\bibliography{bib.bib}

\end{document}